\documentclass[a4paper]{amsart}

\title[Scenario generation in Risk management]{A dynamic approach for scenario generation in \\ Risk management}
\author{J.-P.~Ortega \and R.~Pullirsch \and J.~Teichmann \and J.~Wergieluk}
\address{CNRS, Universit\'e de Franche-Comt\'e, UFR des Sciences et Techniques, 16, route de Gray, 25030 Besan\c{c}on, France; RZB, Am Stadtpark 9, A-1030 Wien, Austria; Vienna University of Technology, Department of Mathematical Methods in Economics, Wiedner Hauptstrasse 8--10, A-1040 Wien, Austria}
\email{Juan-Pablo.Ortega@univ-fcomte.fr, rainer.pullirsch@rzb.at, \newline jteichma@math.ethz.ch, julian.wergieluk@rzb.at}
\thanks{The first and third author gratefully acknowledge support from the START-project Y 328 funded by the Austrian Science Foundation FWF}
\usepackage{amscd}
\usepackage{amsmath}
\usepackage{amssymb}
\usepackage{amsthm}
\usepackage{bbm}
\usepackage{stmaryrd}
\usepackage{eucal}
\usepackage{graphicx}

\numberwithin{equation}{section} \swapnumbers

\newtheorem{satz}{Satz}[section]

\newtheorem{theorem}[satz]{Theorem}

\newtheorem{assumption}[satz]{Assumption}

\newtheorem{definition}[satz]{Definition}

\newtheorem{remark}[satz]{Remark}

\newtheorem{example}[satz]{Example}

\newcommand{\dom}{\operatorname{dom}}

\begin{document}
\begin{abstract}
We provide a new dynamic approach to scenario generation for the purposes of risk management in the banking industry. We connect ideas from conventional techniques -- like historical and Monte Carlo simulation -- and we come up with a hybrid method that shares the advantages of standard procedures but eliminates several of their drawbacks. Instead of considering the static problem of constructing one or ten day ahead distributions for vectors of risk factors, we embed the problem into a dynamic framework, where any time horizon can be consistently simulated. Additionally, we use standard models from mathematical finance for each risk factor, whence bridging the worlds of trading and risk management.

Our approach is based on stochastic differential equations (SDEs), like the HJM-equation or the Black-Scholes equation, governing the time evolution of risk factors, on an empirical calibration method to the market for the chosen SDEs, and on an Euler scheme (or high-order schemes) for the numerical evaluation of the respective SDEs. The empirical calibration procedure presented in this paper can be seen as the SDE-counterpart of the so called Filtered Historical Simulation method; the behavior of volatility stems in our case out of the assumptions on the underlying SDEs. Furthermore, we are able to easily incorporate ``middle-size'' and ``large-size'' events within our framework always making a precise distinction between the information obtained from the market and the one coming from the necessary a-priori intuition of the risk manager.

Results of one concrete implementation are provided. 
\bigskip

\end{abstract}

\maketitle

\bibliographystyle{plain}

\section{Introduction}
A core part of risk management in the banking industry is the identification of risk factors and the generation of scenarios for a one up to ten days time horizon. This task involves dealing with basically three important issues and technical requirements. First, the scenarios need to extract the market information from the available time series of the risk factors. Mathematically speaking, this problem is tackled via the identification of random variables that describe the risk factors one or ten days ahead and that match the most important stylized facts in the time series. A major complication in this classical statistics (or econometrics) question is that the available time series are extremely short (about 500 business days) in comparison with the number of risk factors (several thousands). Even if the time series were long enough, standard statistical methods would take an excessively heavy computational effort to come up with reliable solutions for such a high dimensional problem. 

Second,  the shortness of the time series implies that one has to separately model extreme events, or even middle size events, in order to complement the information provided by the time series and to allow for cautious prediction. The third issue is that the procedure for scenarios generation should be computationally fast since they have to be generated every (business) day in order to be able to evaluate the risk to which all the portfolios of the institution are exposed.

There are two families of ways to deal with the problem of scenario generation in risk management, namely the historical and the distributional approaches. Both strategies produce random variables from which one can sample the scenarios:

\medskip

\noindent {\bf Distributional or Parametric Approach}. The scenarios are constructed by sampling a multivariate parametric distribution that has been calibrated using the historical data. One can also fit certain distributions to vectors of risk factors and then use a copula to put together all these lower dimensional models. There are two major difficulties associated to this method: first, most statistical estimators of the distribution parameters that need to be calibrated yield excessively wide confidence intervals for the estimated coefficients due to the relative shortness of the historical time series in  comparison with the dimensionality of the problem. Second, the available historical data may not correspond to any of the standard frequently used parametric distributions.

\medskip

\noindent {\bf Historical Approach}. The scenarios are obtained by directly sampling the historical data (eventually after rescaling). This method is much privileged in practice due to its ease of implementation and because it does not require any assumption on the historical distributions of the risk factors. This does not mean that this method is hypothesis free; indeed, its statistical legitimacy comes from the so called bootstrap method (see for instance~\cite{bootstrap}), whose requirements are not always satisfied, the most important of them being the iid (independent and identically distributed) character of the risk factors returns. Additionally, this method is obviously incapable of predicting rare (extreme) phenomena that did not happen before in the past; consequently these rare phenomena need to be added ``by hand'' (which is not a disadvantage but simply means that the time series is too short to reflect all possible risks). Several strategies have been proposed to tackle these problems; a particularly sophisticated and powerful one is the so called filtered historical simulation \cite{barbougia:98}, where the historical returns are ``filtered'' by volatility functions calibrated to each time series using GARCH type processes.

\medskip

\noindent {\bf Empirical calibration of SDEs}. The approach that we present in this article is reminiscent of both the historical and the distributional approaches and tries to ``interpolate'' between them. It allows us to properly capture the stylized facts on dependence of the given time series and to easily carry out additional distributional tuning. We create the random variable for future scenarios based on \textbf{stochastic differential equations}, which describe the local dynamics of the respective risk factors (``distributional aspect''). The stochastic differential equations are standard models from mathematical finance like the Black-Merton-Scholes, Cox-Ingersoll-Ross, Heath-Jarrow-Morton models, etc, which are adequately chosen for each individual class of risk factors; in particular all these equations are \textbf{free of arbitrage}. The use of risk management models which could also be used for pricing or hedging derivative contracts is a tremendous conceptual advantage. This way we bridge between the worlds of trading and hedging. In particular the approach allows to actually come up with conceptually coherent methods for risk management which could in principle also be used for trading purposes.

We calibrate the SDEs not via an optimization process as it is customary, but by directly using the historical data to construct local characteristics that capture the stylized facts of the time series. The method relies on the central limit theorem for covariance processes. Under certain conditions that we specify later on in the paper the calibrated SPDEs converge properly if the observed data stem from an SPDE (which is basically the law of large numbers behind the construction). Heavy tails are generated by distinguishing between ``trading time'' and ``clock time'' (randomization of the period of simulation) and by adding extreme events, possibly with heavy tailed distributions. The first effect reflects that the flow of information runs at different speeds under different market volatility conditions and is responsible for ``middle-size'' events; the second effect reflects the possibility of the occurrence of rare (extreme) events not present in the historical data.

Having the filtered historical simulation approach (FHS)~\cite{barbougia:98} in mind, the empirical calibration approach can be seen as the counterpart of FHS on the side of mathematical finance, i.e., instead of using the GARCH-methodology to capture volatility, we apply standard models from mathematical finance and random time changes for the same purpose (see Section \ref{calibration} and \ref{fine_tuning}). To be precise, the issue of capturing volatility correctly is addressed by rescaling with respect to realized volatility parameters of the time series via a sliding window method and by imposing specific assumptions on the dependence of volatility on the underlying risk factors.

\section{Risk factors driven by diffusive stochastic differential equations}\label{basic_setting}

In this section we outline the basic setting of Wiener driven stochastic (partial) differential equations, which underlie the time evolution of risk factors. For stochastic (partial) differential equations we refer to \cite{dapzab:92} as main reference. Details on term structure equations can be found in \cite{filipoteich}.

Our formal setup consists of a filtered probability space $ (\Omega, \mathcal{F}, {(\mathcal{F}_t)}_{t \geq 0}, P) $, where $ {(\mathcal{F}_t)}_{t \geq 0} $ is a complete filtration and $ P $ a probability measure. The basis is assumed to carry a $d$-dimensional standard Brownian motion $ B $. We use as models for the risk factors stochastic differential equations from mathematical finance; this leads us to  stochastic partial differential equations of the following general type: let $ H $ be the Hilbert space where risk factors take values, then we consider
\begin{align}\label{basic_system}
d Y_t & = (\mu_1(Y_t) + \mu_2(Y_t)) dt + \sum_{i=1}^{d} \sigma (Y_t) \bullet \lambda_i dB^i_t, \\
Y_0 & \in H,
\end{align}
where
\begin{equation}
 \sigma(Y) : H_0 \to H_0
\end{equation}
is an invertible linear map on some closed subspace $ H_0 \subset H $ containing the real span  $\langle \lambda_1,\ldots,\lambda_d \rangle $ of the set of \textbf{return directions} $ \lambda_1, \ldots, \lambda_d \in  H$ (see Remark \ref{returns} for the use of the word return here). Such volatility structures are usually called \textbf{constant direction volatilities}, see for instance \cite{slinko}. The ``volatility factor'' $ \sigma $ is chosen appropriately for the respective risk factors in order to exclude immediate arbitrages (like, e.g., negative interest rates). It should be interpreted as an \textbf{a priori given} factor governing the shape of the support of the risk factors;  it is hence a \textbf{geometric factor}. The vector field $ \mu_1 $ corresponds to appropriately chosen no-arbitrage conditions;  $ \mu_2 $ is another vector field  that lies in the span  $\langle \lambda_1,\ldots,\lambda_d \rangle $ and corresponds to an appropriate Girsanov type change of measure.

We suppose that the usual conditions on the drift $ \mu_1 $ hold, that is, it can be written as
\begin{equation}
 \mu_1(Y) = \mathcal{A} Y + \mu_3(Y) 
\end{equation}
for some $ Y \in \dom(\mathcal{A}) $, where $ \mathcal{A} $ denotes the generator of a strongly continuous semigroup which  in our case will always be a shift to the right or the identity semigroup on the Hilbert space of risk factors $ H $. The maps $ \mu_3, \mu_2 : H \to H $ are smooth and all their derivatives are bounded ($C^{\infty}$-bounded map); they represent the \textbf{no arbitrage conditions}. Certainly -- in view of existence and uniqueness -- we have to assume global Lipschitz properties for the respective vector fields.

\begin{example}
\normalfont
Our driving example will be a joint interest rate (IR) and foreign exchange (FX) market for $ p+1 $ countries with different currencies and interest rate levels ($p \geq 1$). The risk factors of this market are the different log-exchange rates and forward interest rates. For the basic equations of this market we refer to \cite{slinko}. Denote by $ P^{0}_t(x) $ and $ P^{\alpha}_t(x) $, $ \alpha \geq 1 $ the time $t$ prices of the domestic and foreign default-free, zero coupon bonds, where $ \alpha=1,\ldots,p $ denotes the different economies where exchange rates and zero coupon bonds are observed. The bonds are maturing at $ T = t + x \geq 0 $. We define the domestic and foreign instantaneous forward rates as
\begin{equation}
r^{0}_t = \partial \log P^{0}_t, \quad r^{\alpha}_t = \partial \log P^{\alpha}_t.
\end{equation}
Using the Musiela parametrization and the standard Heath-Jarrow-Morton drift condition we can write down the domestic forward rate dynamics under the domestic martingale measure
\begin{equation}
dr^0_t  =  \bigr( \frac{d}{dx} r^0_t  + \sum_{i=1}^d    \sigma (r^0_t) \bullet \lambda_i \, \int \sigma (r^0_t) \bullet \lambda_i \bigl )dt + \sum_{i=1}^{d} \sigma (r^0_t) \bullet \lambda_i \, dB^i_t,
\end{equation}
Analogously under the foreign martingale measure we obtain a Heath-Jarrow-Morton equation for the foreign forward rate dynamics.
If the dynamics of the spot log-exchange rate process with respect to the domestic martingale measure is as follows
\begin{equation}\label{fx_equation}
d S^{\alpha}_t = (r^{0}_t (0) - r^{\alpha}_t (0))dt - (\sum_{i=1}^d \frac{{(\delta^{\alpha}_i)}^2}{2})dt + \sum_{i=1}^d \delta^{\alpha}_i dB^i_t,
\end{equation}
we obtain under the domestic martingale measure the foreign forward rate dynamics as
\begin{align}\label{ir_equation}
dr^{\alpha}_t & =  ( \frac{d}{dx} r^{\alpha}_t  +  \sum_{i=1}^d  \sigma (r^{\alpha}_t) \bullet \lambda_i \int \sigma (r^{\alpha}_t) \bullet \lambda_i  - \\ \nonumber
& - \sum_{i=1}^d    (\sigma (r^{\alpha}_t) \bullet \lambda_i) \delta^{\alpha}_i )dt +  \sum_{i=1}^{d} \sigma (r^{\alpha}_t) \bullet \lambda_i \, dB^i_t,
\end{align}
for $ \alpha = 1,\ldots, p $. Recall the no-arbitrage conditions, namely, the processes
$$
\exp(-\int_0^t r^{0}_t (0) dt) P^{0}_t(T), \quad \exp(-\int_0^t r^{0}_t (0) dt) \exp(S^{\alpha}) P^{\alpha}_t(T)
$$
are (local) martingales on $ [0,T] $ for $ \alpha =1,\ldots,p $ and $ T \geq 0 $.

This system of stochastic partial differential equations (SPDEs) is of the form \eqref{basic_system}, when we consider a vector of risk factors $ Y = (r^0,r^1,\ldots,r^p,S^1,\ldots,S^p) $ forming an element of a large Hilbert space
\begin{equation}
H := H^0\times \cdots \times H^{p} \times \mathbb{R} \times \cdots \times \mathbb{R}.
\end{equation}
By abuse of notation we use the same letters $ \lambda_1,\ldots,\lambda_d $ and $ \sigma $ for the respective fields on the Hilbert space of risk factors, such that we obtain an equation of type \eqref{basic_system}. Notice that we can easily perform a change of measure in order to attain the physical measure. This corresponds then to introducing an appropriately modified drift, denoted by the additional term $ \mu_2 (Y) $.
\end{example}

\section{Calibrating the diffusive SPDEs to the market}\label{calibration}

In this section we propose a calibration to (diffusive parts of) the historical risk factor time series, which does not need neither optimization procedures nor the numerical evaluation of the model in question. Instead of using those standard approaches, which pose serious problems in high dimensions, we will rather consider a ``layman's calibration'' of the model, where the characteristics of the SPDE are deterministically  and directly calculated from the observed time series. This corresponds to the ``historical approach'', which is standard for several branches of scenario generation, however, we only apply the historical approach at an infinitesimal level (thinking of a time-tick, usually from hours up to days in scenario generation, as the infinitesimal element): the global dynamics is then calculated by means of a stochastic differential equation whose local features are determined by the time series.

We suppose that we are in  the setting of Section \ref{basic_setting} and we assume a time series, i.e. a sample of equation \eqref{basic_system}, on equidistant grid points of distance $ \Delta $, denoted by $ Y_1,\ldots,Y_K $. Our goal consists of estimating the volatility directions $ \lambda_1,\ldots, \lambda_d $ out of the observations $ Y_1,\ldots,Y_K $ in a simple way (we regard $ Y_K $ as the last observation, which appeared one time tick $ \Delta $ ago). We do really aim for an estimation in the proper sense, i.e., we would like to have a limit theorem for SPDEs that relates the number of observations $ K $, the time-distance $ \Delta $ of two contiguous observations, and the number of volatility directions $ d $. We certainly do not want the dimension of the Hilbert space of risk factors to enter into the construction, since it might be infinite. 

For later use and for the sake of simplicity we state the stochastic partial differential equation of type \eqref{basic_system} calibrated to the given sample $ Y_1,\ldots,Y_K $
\begin{align}\label{calibrated_equation}
dX^{(K)}_t = & (\mu^{(K)}_1(X^{(K)}_t) + \mu^{(K)}_2(X^{(K)}_t) )dt + \\ \nonumber
&  + \frac{1}{\sqrt{\Delta(K-1)}} \sum_{i=1}^{K-1} {\sigma(X^{(K)}_t)} \bullet ({\sigma(Y_i)}^{-1}{(Y_{i+1}-Y_i)}) \, dW^i_t, 
\end{align}
where $ \sigma $ is a  known, non-vanishing geometric function on the risk factors describing the local dynamics. We specify in the sequel the necessary conditions for an appropriate limit theorem. We denote by $ X_t ^{(K)}$ the dependence of the solution process $ X_t $ on the number $ K $ of observations (which  also determines the distance $\Delta$ between two observations). The calibrated equation is defined on an independent stochastic basis with a $ K-1 $ dimensional Brownian motion $ (W^1,\ldots,W^{K-1}) $.

Notice here that $ \mu_1^{(K)} $ is specified by no-arbitrage conditions from the volatility structure, as given in the equations \eqref{fx_equation} and \eqref{ir_equation}. Furthermore we have to assume the following technical assumption, which allows to construct strong solutions and not only mild (respectively weak) ones. For all necessary details see \cite{filipoteich} and \cite{bautei}.

\begin{assumption}\label{ass_lambda}
 The constant volatility directions $ \lambda_1,\ldots,\lambda_d $ are elements of $ \dom(\mathcal{A}^{\infty}) $, which is the domain where all powers of $ \mathcal{A} $ are defined (see for instance \cite{bautei}. Furthermore $ \sigma: \dom(\mathcal{A}) \to \dom(\mathcal{A}) $ is assumed to be a Lipschitz map.
\end{assumption}

\begin{theorem}
Consider equation \eqref{basic_system} where $ \sigma $ is a given map and $ Y_0 \in \dom(\mathcal{A}) $ a given initial value, but the directions $ \lambda_1,\ldots,\lambda_d $ are unknown. We assume Assumption \ref{ass_lambda}. We collect a time series of observations $ Y_1,\ldots,Y_K $ on an equidistant time grid of width $ \Delta $ that cover an interval of length $ T=K \Delta $. Refining the observations by making  $ \Delta = \frac{T}{K} $ smaller and smaller leads to the following convergence statement:
\begin{equation}
\label{main convergence result}
\lim_{K\to \infty} X^{(K)}_t = Y_t 
\end{equation}
in distribution for any $ t \geq 0 $ if $ X_0 = Y_0 $. The underlying limit result is the following Gaussian one,
\begin{align}
& \lim_{K\to \infty} \int_0^t {\sigma(X^{(K)}_t)}^{-1} dX^{(K)}_t - \int_0^t {\sigma(X^{(K)}_s)}^{-1}(\mu_1^{(K)}(X^{(K)}_s) + \mu_2^{(K)}(X^{(K)}_s) ds \label{gaussian_simulated_equation} \\ 
= &\lim_{K\to \infty} \frac{1}{\sqrt{\Delta(K-1)}} \sum_{i=1}^{K-1} ({\sigma(Y_i)}^{-1}{(Y_{i+1}-Y_i)}) W^i_{t} \label{gaussian_simulated_equation_explicit_form} \\ 
= & \int_0^t {\sigma(Y_t)}^{-1} dY_t - \int_0^t {\sigma(Y_s)}^{-1}(\mu_1(Y_s) + \mu_2(Y_s)) ds. \label{gaussian_original_equation}
\end{align}
\end{theorem}

\begin{proof}
Let $ H $ be the risk factor Hilbert space and $ \lambda \in H $. Denote by
$
\lambda \otimes \lambda
$  the bilinear operator defined by $ ( \lambda \otimes \lambda )(v,w) := \langle \lambda , v \rangle \langle \lambda, w \rangle $, for $ (v,w) \in H \times H $. Linear combinations of these (one-dimensional) bilinear operators are dense in the set of Hilbert-Schmidt operators (with respect to the Hilbert-Schmidt norm). Let $ M $ be a Hilbert space valued It\^o process given by
$$
dM_t = a_t dt + \sum_{i=1}^d b_t^i dW^i_t,
$$
for a family $ a,b^1,\ldots,b^d $ of square-integrable processes; then, the quadratic covariation
$$
\sum_{i=1}^d \int_0^T (b_t^i \otimes b_t^i) \, dt
$$
is approximated with respect to the Hilbert-Schmidt norm by
$$
\sum_{j=0}^{K-1} \left(M_{\frac{T(j+1)}{K}} - M_{\frac{Tj}{K}}\right) \otimes \left(M_{\frac{T(j+1)}{K}} - M_{\frac{Tj}{K}}\right),
$$
see for instance \cite{dapzab:92}. The Gaussian random variable in equation~(\ref{gaussian_simulated_equation_explicit_form}) has a covariance matrix
\begin{equation}
\frac{t}{\Delta(K-1)} \sum_{i=1}^{K-1} ({\sigma(Y_i)}^{-1}{(Y_{i+1}-Y_i)}) \otimes {({\sigma(Y_i)}^{-1}{(Y_{i+1}-Y_i)})}
\end{equation}
which precisely corresponds, up to the factor $ t/\Delta$, with the standard statistical estimator  the returns (see Remark \ref{returns} for the use of the word return here) of the sample $ Y_1,\ldots,Y_K $ multiplied by $ {\sigma(Y_i)}^{-1} $. On the other hand, this estimator converges, as $ K $ tends to infinity, to 
\begin{equation}
\lim_{K \to \infty} \frac{1}{\Delta(K-1)} \sum_{i=1}^{K-1} ({\sigma(Y_i)}^{-1}{(Y_{i+1}-Y_i)}) \otimes {({\sigma(Y_i)}^{-1}{(Y_{i+1}-Y_i)})} = \sum_{i=1}^d \lambda_i \otimes \lambda_i.
\end{equation}
This is due to the fact that the quadratic covariation of the stochastic process
\begin{equation}
{\sigma(Y_t)}^{-1}dY_t = {\sigma(Y_t)}^{-1}(\mu^{(K)}_1(Y_t) + \mu_2(Y_t)) dt + \sum_{i=1}^d \lambda_i d B^i_t 
\end{equation}
is given by $ \sum_{i=1}^d (\lambda_i \otimes \lambda_i) \, dt $, whence $ \frac{1}{T} \sum_{i=1}^d (\lambda_i \otimes \lambda_i) $ is approximated with respect to the Hilbert-Schmidt norm by
$$
\frac{1}{\Delta(K-1)} \sum_{i=1}^{K-1} ({\sigma(Y_i)}^{-1}{(Y_{i+1}-Y_i)}) \otimes {({\sigma(Y_i)}^{-1}{(Y_{i+1}-Y_i)})}.
$$
Notice here that $ \Delta K = T $. This proves the equality of \eqref{gaussian_simulated_equation} and \eqref{gaussian_original_equation}. 

In order to prove the convergence result~(\ref{main convergence result}) we will have to apply the following limit theorem in distribution for a sequence of stochastic partial differential equations: let $ X^{(K)}_t$ and $ Y_t $ be the unique solutions of the equations
\begin{align}
dX^{(K)}_t & = (\mu^{(K)}_1(X^{(K)}_t) + \mu_2^{(K)}(X^{(K)}_t) )dt + {\sigma(X^{(K)}_t)} \bullet d B^{(K)}_t, \; X^{(K)}_0 = X_0,\\
dY_t & = (\mu_1(Y_t) + \mu_2(Y_t) )dt + {\sigma(Y_t)} \bullet d B_t, \; Y_0 = X_0,
\end{align}
where
\begin{align}
B^{(K)}_t & = \sum_{i=1}^{K-1} {\sigma(Y_i)}^{-1}{(Y_{i+1}-Y_i)} \, W^i_t, \\
B_t & = \sum_{i=1}^d \lambda_i \, B^i_t
\end{align}
for $ t \geq 0 $. Then 
\begin{equation}
\label{limit theorem we need}
 X^{(K)}_t \longrightarrow  Y_t  
\end{equation}
in distribution as $ K \to \infty $. In order to show~(\ref{limit theorem we need}), we start by noticing that, due to the previous consideration we have
$$
B^{(K)} \xrightarrow[K \to \infty]{} B
$$
in distribution on the Hilbert space $ H $. The limit theorem~(\ref{limit theorem we need}) can be concluded from the stability theorems in \cite{filtaptei} and the uniqueness in law established in \cite{dapzab:92}: first, by uniqueness in law, the actual Brownian motion in use is not relevant. Hence, we can interpret the convergence of $ B^{(K)} $ as $ K \to \infty $ with respect to independent Brownian motions $G^l_t $, each with a certain volatility $ \kappa^{(K)}_l$ and associated to a certain eigendirection $e_l^{(K)} $, that is, we can rewrite in law
$$
B_t^{(K)} = \sum_{i=1}^\infty \kappa^{(K)}_l e_l^{(K)} G^l_t.
$$
This yields, for each $ l $, vector fields converging to a limit as $ K \to \infty $, namely the eigensystem of $ \sum_{i=1}^d \lambda_i \otimes \lambda_i $, if one carefully enumerates the volatilities $ \kappa_l^{(K)} $ in increasing order and keeps track of multiple eigenvalues. Second, if we reformulate the stochastic partial differential equations \eqref{basic_system} and \eqref{calibrated_equation} with respect to the sequence of independent Brownian motions $ G^l $, we find ourselves with a sequence of vector fields converging to well-specified limits. Finally, the required statement follows from the stability result in Proposition 9.1 of \cite{filtaptei}.
\end{proof}

\begin{remark}
\normalfont
The methodology in the previous theorem is related to the one in \cite{malmanrec} in the sense that we read from the quadratic variation process and its approximations the relevant information for the a posteriori construction of a process which is close in law to the originally observed one. Indeed, in \cite{malmanrec} the quadratic variation process and its approximations are -- via harmonic analysis methods -- investigated in order to reveal information on the non-linear volatility vector fields. In contrast, we do assume knowledge on the type of non-linearity in the equation and tackle the problem of determining the volatility directions $ \lambda_1, \ldots, \lambda_d $. The convergence results based on Wiener's theorem, as stated in \cite{malmanrec}, could be similarly applied.
\end{remark}

\begin{remark}
\normalfont
Notice that we do not need to calculate empirical covariance matrices of the time series in order to to capture the local correlation structure of \textbf{all} risk factors in the calibrated equation \eqref{calibrated_equation}.
\end{remark}

\begin{remark}
\normalfont
Up to some technical complications one could also consider $ d = \infty $ here.
\end{remark}

\section{Rare events and Fine-tuning through random time changes}\label{fine_tuning}

Besides the calibration of the risk factor SPDEs proposed in equation \eqref{calibrated_equation} we have to take into account stochastic volatility. We express stochastic volatility in our setting through the appearance of random times, which can be nicely interpreted as difference between physical time and trading time. As the spirit of this article is to embed a static sampling problem (in contrast to the distributional approach where only two random variables are described) into a dynamic problem defined by stochastic (partial) differential equations, we also embed the random time into a random time change, which can then a fortiori be observed through an stochastically changing volatility process. This approach is a continuous time version of GARCH-approaches, which are applied in filtered historical simulation in risk management (see \cite{barbougia:98}). Here, in contrast, we ``calculate'' stochastic volatility along the time series as realized volatility and rescale the time series by this realized volatility before we apply those rescaled (``filtered'') returns to calibrate the underlying stochastic differential equations as proposed in equation \eqref{calibrated_equation}

In this section we therefore introduce the additional methodology, which allows us to include these changing volatility effects  of the observed diffusion processes and which permits jumps. Both phenomena, stochastic volatility and jumps, are crucial to capture important stylized facts of financial time series. For instance, when we model rare (extreme) events we those events as additional jumps added to the diffusive setting of equation \eqref{basic_system}. Recall that these additional jumps also change the drift $ \mu_1 $ due to no arbitrage conditions. In this section we assume the existence of a basis $ \{e_j\} $ of the respective Hilbert spaces of risk factors in order to speak about the risk factors $ Y^j $, which denote the projection of $Y $ onto the $j$-th basis element.

\begin{remark}\label{returns}
\normalfont
We shall apply in the sequel the notion \emph{return} for two slightly different concepts. When considering a time series $ Y_1,\ldots,Y_K $, which is a collection of vectors, then we shall call any difference $ Y_{i+1} - Y_i $ an (observed) return (of risk factors), where we should in fact speak of a vector of returns $ {(Y_{i+1}^j -Y_i^j)} $ of single risk factors $ Y^j $. Running indices for risk factors are always denoted by $ j $, the index along the time series is usually $ i $. The length of the time series will be always be $ K $.
\end{remark}

We define rare events using known time series of the risk factors and by the risk manager's very own intuition. Let $ Y_1,\ldots,Y_K $ denote a time series of risk factors. We form a time series of returns $ Y_{i+1} -Y_i $ for $ i=1,\ldots,K-1 $ from it and estimate the empirical (co-)variance $ \widehat{\sigma} $ among its different components. Additionally we define local extractions of the time series of length $ L $, i.e. $ Y_{k+i+1} - Y_{k+i} $, for $ k = 1,\ldots,L $ and for $ i = 0,\ldots, K-L-1 $. The empirical (co-)variance of this extraction is denoted by $ \widehat{\sigma}^{(i+L)} $ and corresponds to the estimation of the volatility looking back $ L $ days, usually called (local) realized volatility.

We now consider a diffusive equation of type \eqref{basic_system} with jumps and random time change. More precisely, let $ \tau^j_t = \int_0^t A^j_s ds $ be a random time change for risk factor $ j $, i.e. $ A^j $ is a non-negative, locally integrable previsible stochastic process (with respect to the natural filtration of the Brownian and Poissonian componentes), whence  $ t \mapsto \tau^j_t $ is an almost surely increasing random time. Then the equation to be calibrated looks like
\begin{align}
d Y^j_{\tau^j_t} & = (\mu_1(Y_{\tau_t}) + \mu_2(Y_{\tau_t})) d\tau^j_t + \sum_{i=1}^{d} {( \sigma (Y_{\tau_t}) \bullet \lambda_i)}^j dB^i_{\tau^j_t} + dN_{\tau^j_t} ,\\
d \tau^j_t & = A^j_t dt, \\
Y_0 & \in H,
\end{align}
where $ N $ denotes a Poisson process with values in $ H $ and jump rate $ \mathcal{r} $ calibrated from the time series. The solution of this random SPDE can be written through the random time change by solving \eqref{basic_system} up to the time $ \tau^j_t $ for the $j$-th risk factor. It is possible to estimate the volatility levels $ A^j_{t} $ from the market up to an (a priori unknown) factor by observing a proper quadratic variation process between two jumps at time $ a $ and $ b $, namely
$$
	\langle \int {{(\sigma(Y_t)}^{-1} \bullet dY_t)}^j \rangle|_a^b 
= 
\int_a^b A^j_s ds \, \sum_{i=1}^d \int_a^b {( \lambda_i^j)}^2 ds. 
$$

As a second notation we introduce the ratios of local empirical volatility and today's local volatility, namely
$$
\widehat{R}^j_i = \frac{\sqrt{\widehat{\sigma}^{(i)}_{jj}}}{\sqrt{\widehat{\sigma}^{(K-1)}_{jj}}},
$$
which is defined for each risk factor $ j $ separately. We denote this quantity by $ \widehat{R}^j_i $ at time $ i $ in the time series. Notice that this quantity is well-defined, even under the regime of a random time change, since the (unknown) proportionality factor cancels out.
\begin{definition}\label{rescaled_returns}
Let $ Y_1,\ldots,Y_K $ denote a time series of risk factors. We form a time series of returns $ Y_{i+1} -Y_i $ for $ i=1,\ldots,K-1 $ from it and denote by
$$ 
\frac{(Y^j_{i+1} - Y^j_i)}{\widehat{R}^j_i} 
$$
for risk factor $ j $ and $ i =L,\ldots,K-1 $ the rescaled returns of the time series.
\end{definition}

\begin{remark}
\normalfont
In order to make the choice of extreme events more stable along time, one can choose $L$ for re-scaling the returns and for extracting extreme events independently. In the subsequent example this is the case.
\end{remark}

In order to calibrate the equation in the spirit of \eqref{calibrated_equation} we start from the set of rescaled returns and separate from it the set of extreme returns.

\begin{definition}
A return at index $ L \leq i \leq K-1 $ is said to be extreme at level $ \eta $ with $ M $ violations if 
$$ 
\eta \times \sqrt{\widehat{\sigma}^{(K-1)}_{jj}} \leq \frac{| Y^j_{i+1} - Y^j_i |}{\widehat{R}^j_i} \,
$$
holds true for at least $ M $ different risk factors $ j $. The set of all indices belonging to extreme events at level $ \eta $ with $ M $ violations is denoted by $ \mathcal{E}_{\eta,M} $ and is a subset of the set of all indices $ \mathcal{K} = \{1,\ldots,K-1\} $.
\end{definition}

By means of the set $ \mathcal{E}_{\eta,M} $ we can now define a jump structure, where with a certain jump rate and jump measure jumps from $ \mathcal{E}_{\eta,M} $ are mixed (in an independent fashion) into the diffusion process calibrated with returns from the set $ \mathcal{K} \setminus \mathcal{E}_{\eta,M} $. 

For the purpose of numerical efficiency of the simulation we shall most of the time neglect the fact that the random time change depends on each risk factor $Y^j$. Actually the simulation time $ \tau_{t} $ is chosen with respect to the distribution of $ \int_0^{t} A_s ds $, which is estimated from an averaged realized variance structure.

\section{Evaluating the stochastic differential equations}

In our case the parameters of the SPDE are specified by the the empirical calibration method outlined in Section \ref{calibration} and by an a priori choice of the geometric term $ \sigma $. Therefore, in order to fully characterize the equation, the following points need to be taken into account:
\begin{enumerate}
 \item Determine the sets $ \mathcal{E}_{\eta,M} $ and $ \mathcal{K} \setminus \mathcal{E}_{\eta,M} $ by fixing a level $\eta$ at which events are considered as extreme and a number $ M $ of violations.
 \item The number of involved Brownian motions is the cardinality of $\mathcal{K} \setminus \mathcal{E}_{\eta,M}$.
 \item The volatility directions are given by rescaled returns using equation \eqref{calibrated_equation}.
 \item The semigroup $\mathcal{A}$ is the shift semigroup for term structure of the risk factors, otherwise the identity.
 \item $ \sigma $ is chosen a priori according to the (geometric) structure of the problem, i.e., some risk factors have to stay positive (like interest rates), others are free (like log-prices).
\end{enumerate}

For the evaluation of the SPDE \eqref{calibrated_equation} we propose an explicit-implicit Euler scheme (see for instance \cite{filtaptei}), which has the following one-step Euler structure:
\begin{itemize}
 \item Prepare an initial risk factor $ Y_0 $.
 \item Trigger a jump (with a specified jump rate), or not. If a jump is triggered, then sample its size using the jump measure and add it to $ Y_0 $.
 \item Choose Brownian increments $  B^i_t$ with variance $ t $ along the time increment $ t $.
 \item Evaluate
\begin{align}
S_{t} Y_0 + & \mu^{(K)}_3(Y_0)t + \mu^{(K)}_2(Y_0) t + \\
+ & \frac{1}{\sqrt{\Delta(K-1)}}\sum_{i \in \mathcal{K} \setminus \mathcal{E}_{\eta,M}} \sigma(Y_0) \bullet {\sigma(Y_i)}^{-1}\frac{(Y_{i+1}-Y_i)}{\hat{R}_i^.} \, B^i_t.
\end{align}
 \item Terminate the Euler step.
\end{itemize}

\begin{remark}
\normalfont
In our concrete implementations the choice of $ \sigma $ depends on the respective risk factors. For instance, for interest rates we have chosen fields of linear operators depending via square-root like functions on the level of yield at certain maturities.
\end{remark}

\section{Results from a concrete implementation}

In this last section we present results from a concrete implementation of a scenario generator for five foreign currencies (CHF, GBP, JPY, USD, ZAR) and six yield curves, i.e., five foreign yield curves and one domestic yield curve (Euro zone). For this purpose we use market data from March 2006 to March 2009, including in particular the highly volatile events of the last quarter of 2008. In total, this risk management problem consists of $ 82 $ risk factors, even though more points on the yield curve could have been evaluated. In a given day, the available historical time series has a length of $ K = 500 $ business days back in time, i.e., in March 2009 we look back up to March 2007, for instance. The implementation is based on Section \ref{basic_setting} and the systems of equations presented therein. In the light of Sections \ref{calibration} and Section \ref{fine_tuning} the realized variances (with $L=20$) and the set of extreme events (for the calculation of extreme events we used $L=40$) are calculated: the number of violations $ M = 4 $ and the level $ \eta = 4 $ yield a number of extreme events of $32$ (for $M=1$, i.e.~an extreme event occurs if we have one violation for one risk factor, we obtain $162$ extreme events). The jump rate is chosen to be $ 2 $ percent (yielding about $ 100 $ jump events in $ 5000 $ scenarios). The random time change has been chosen ad-hoc to be $ 0.9 $ days in $ 90 $ percent of the cases and $ 1.9 $ days in $ 10 $ percent of the cases for the sake of simplicity. A precise calibration might yield even better results.

The following pictures show several single distinguished historical and generated distributions of risk factors. The historical distributions are rescaled by $ R $ as described in Definition \ref{rescaled_returns}. The last figure shows the backtesting results for a randomly chosen roll-over portfolio of linear FX- and IR-instruments, i.e., the times to maturity of instruments in the portfolio are kept constant along the backtesting. The distributional plots show on the left hand side the historical distribution of log-returns for one day and on the right hand side the simulated distributions of log-returns for one day. 

It is apparent that our simulated scenarios resemble the historical distributions by smoothening them in a consistent way along all risk factors and by preserving extreme risks. We obtain backtesting results of four VAR-violations and two portfolio-loss-values beyond the level of expected shortfall in 250 days at a $99\% $ confidence level.

In the first set, the smoothening  is the main visible effect  in the simulation results. Also other stylized facts of the distribution like skewness and kurtosis are preserved.
\begin{center}
 \includegraphics[width=10 cm,height=7 cm]{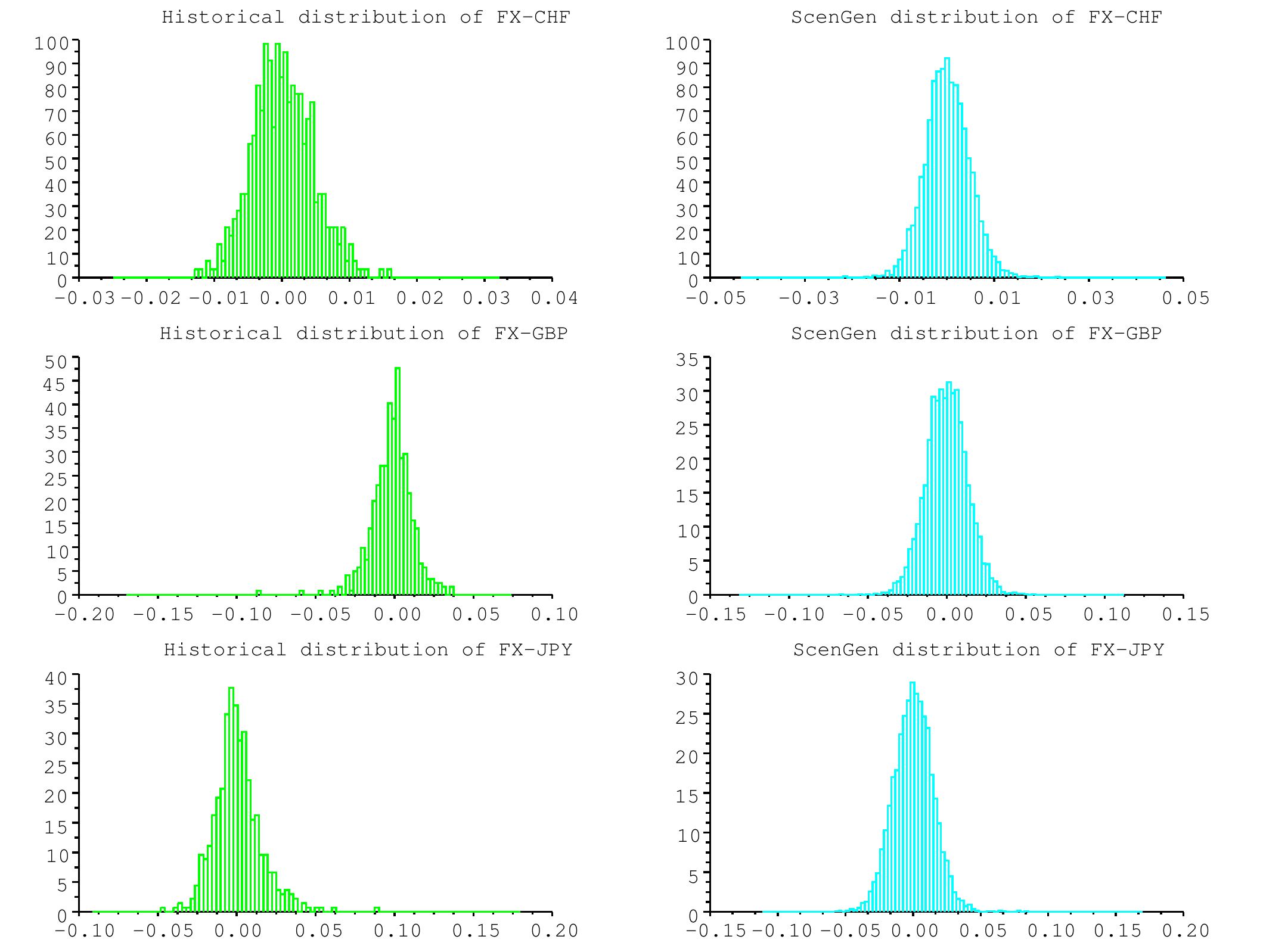}
\end{center}
\begin{center}
 \includegraphics[width=10 cm,height=7 cm]{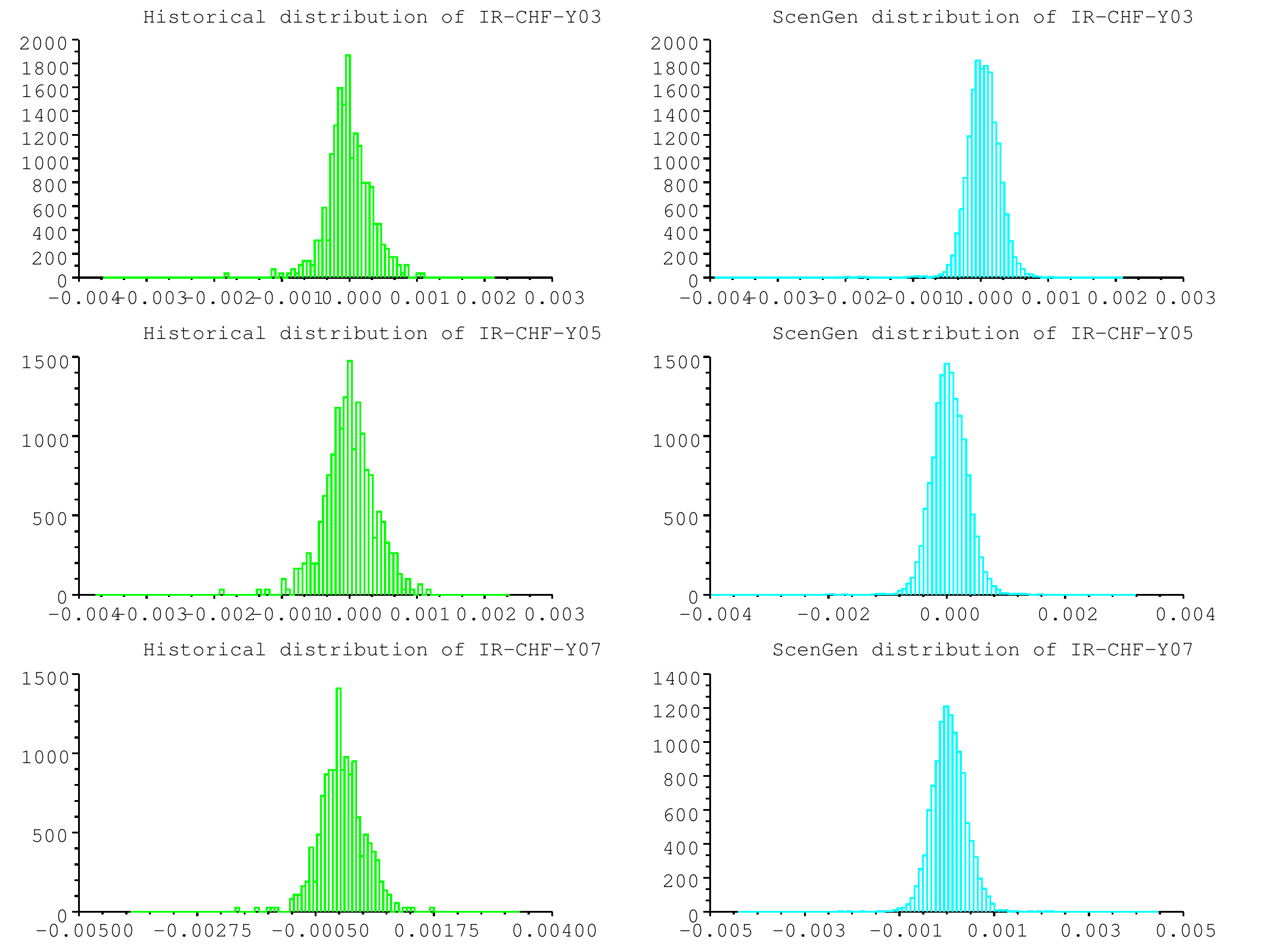}
\end{center}
In the second set of plots, mass of the historical distributions is transported from the center to middle-size events.
\begin{center}
 \includegraphics[width=10 cm,height=7 cm]{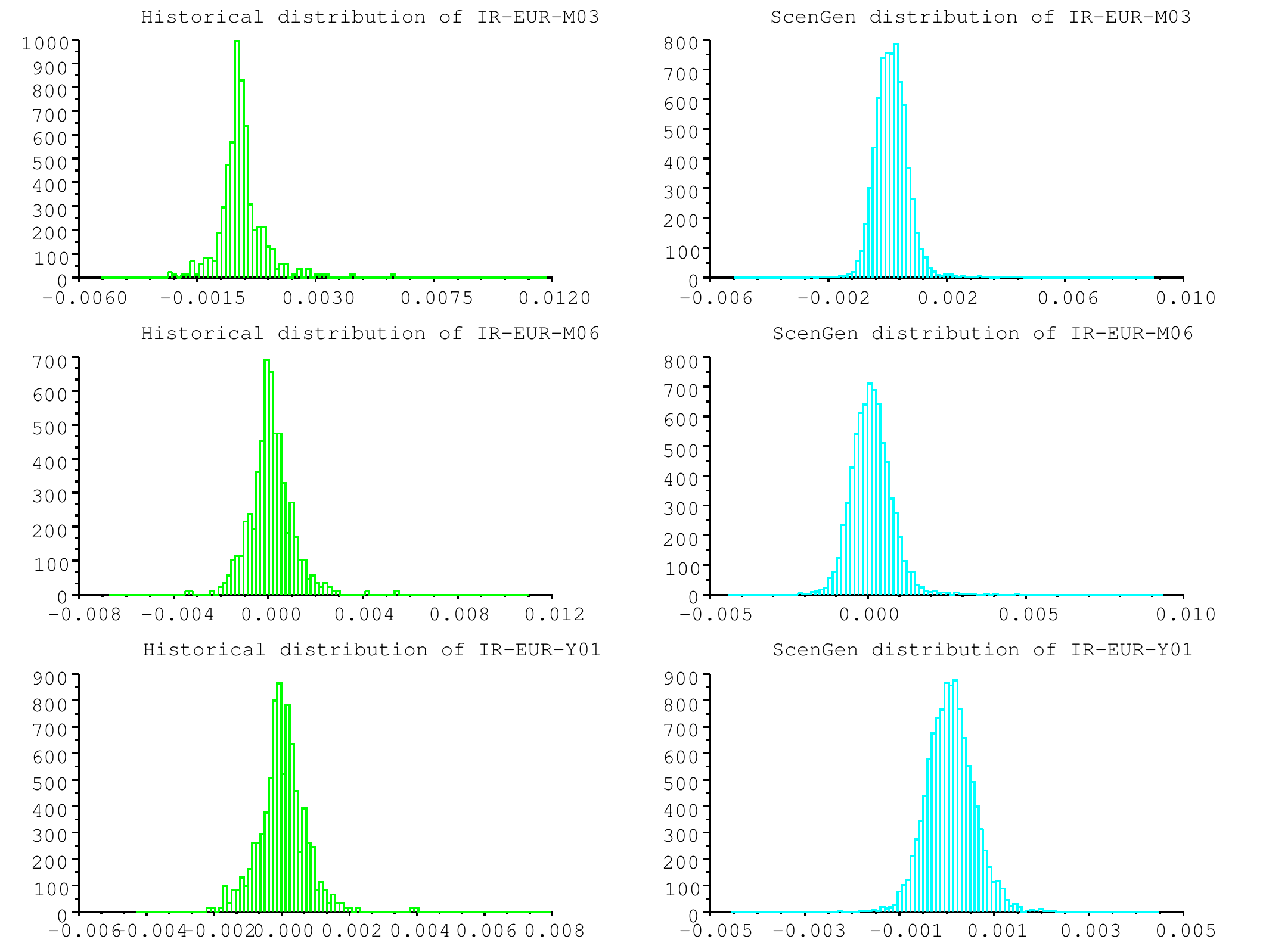}
\end{center}
In the third set of plots the mass transport effect is seen in an even more extreme way, since for ZAR a lot of returns have historically been zero. In this sense the constructed scenario generator also has a ``repairing'' effect on data.
\begin{center}
 \includegraphics[width=10 cm,height=7 cm]{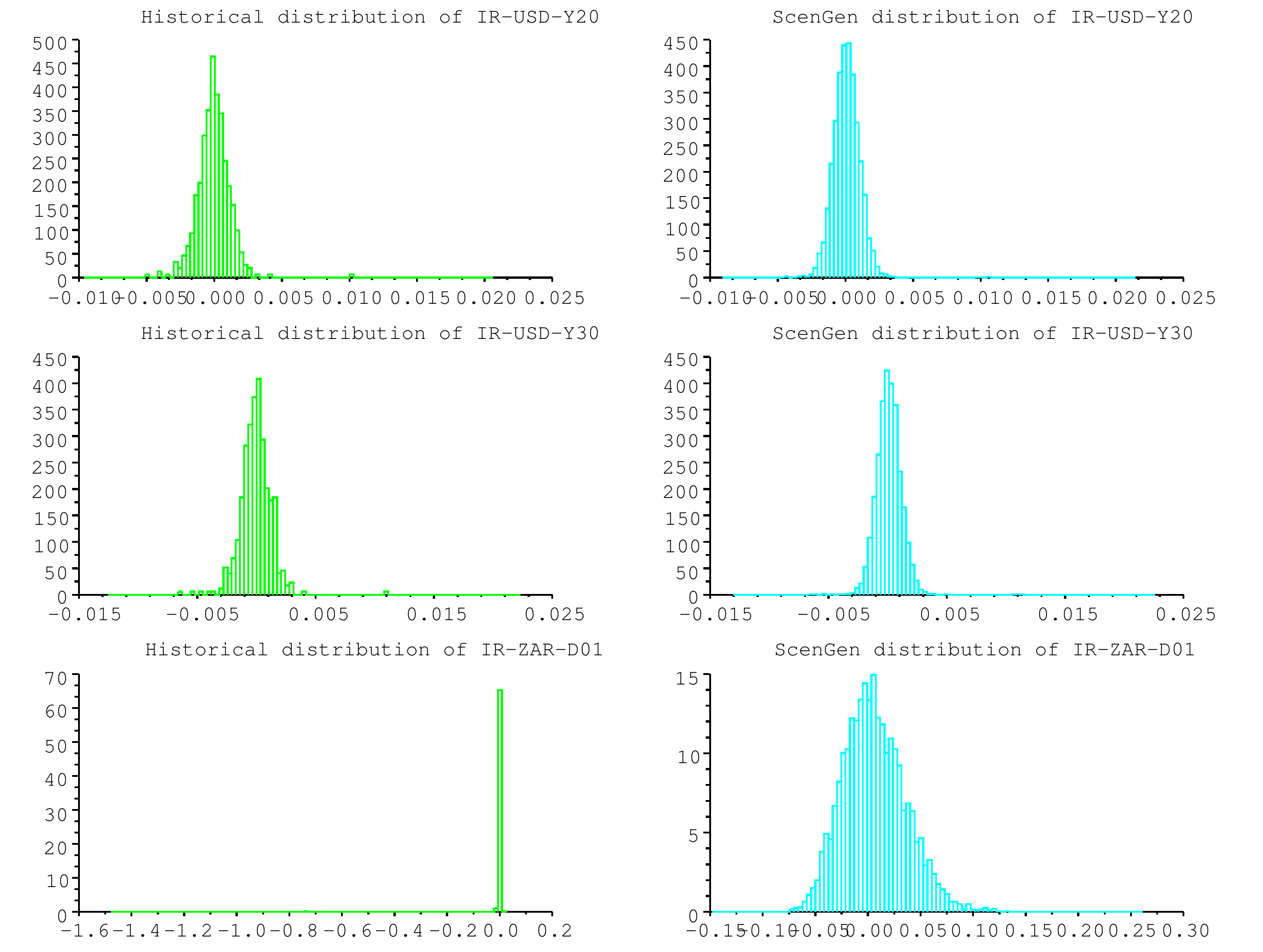}
\end{center}
\bigskip
Finally, we present backtesting results of  a two roll-over portfolios during $250$ business days at the $99\% $ confidence level. Returns are in red, VaR in blue and the value of expected shortfall is in green. One can clearly see the increase of the level of VAR (and expected shortfall) during the events of the last quarter 2008. We have omitted the scaling on the y-axis as violations are only a level of relative size.
\begin{center}
\label{fig_1}
 \includegraphics[width=12 cm,height=11 cm]{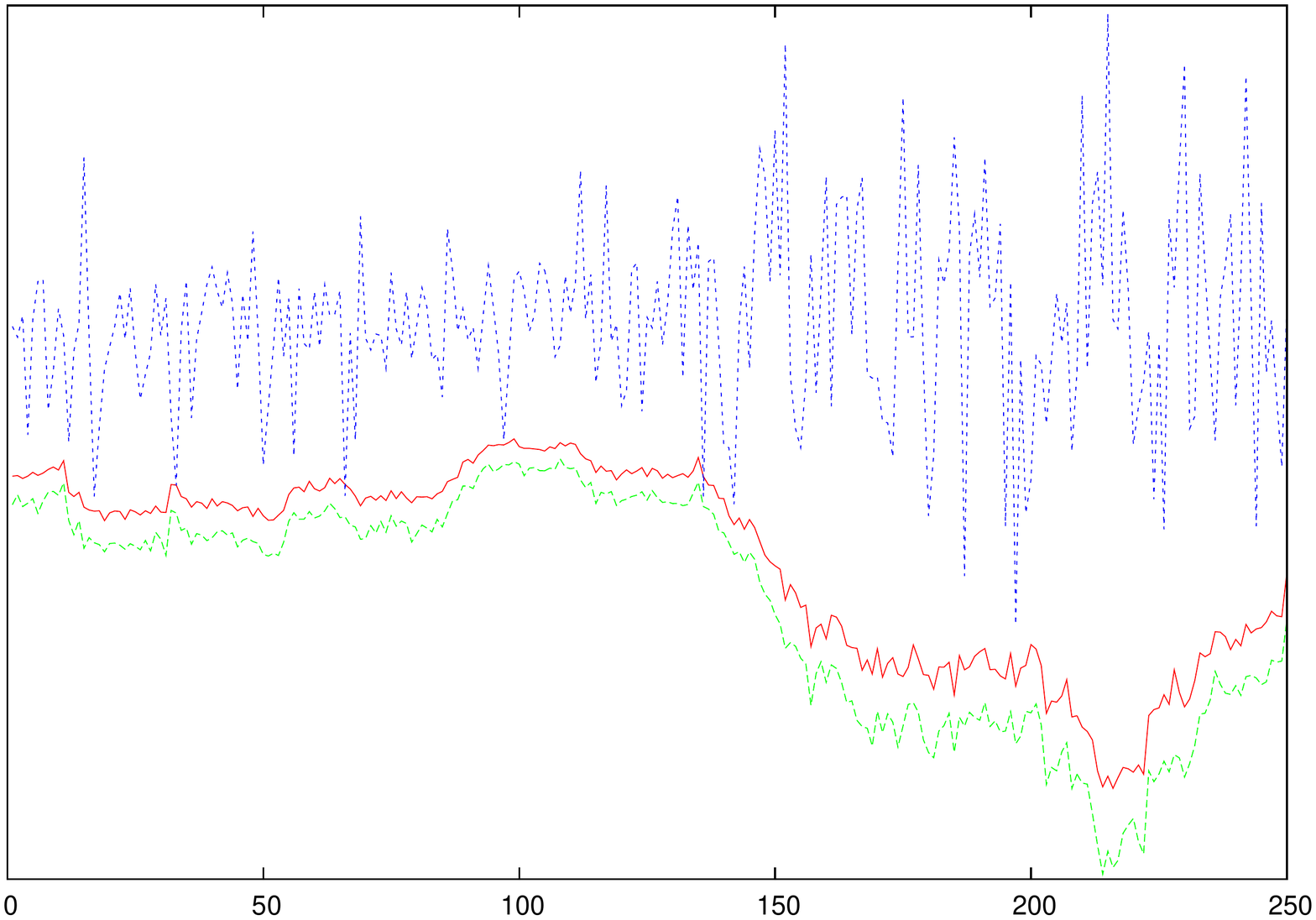}
\end{center}

\begin{center}
\label{fig_2}
 \includegraphics[width=12 cm,height=11 cm]{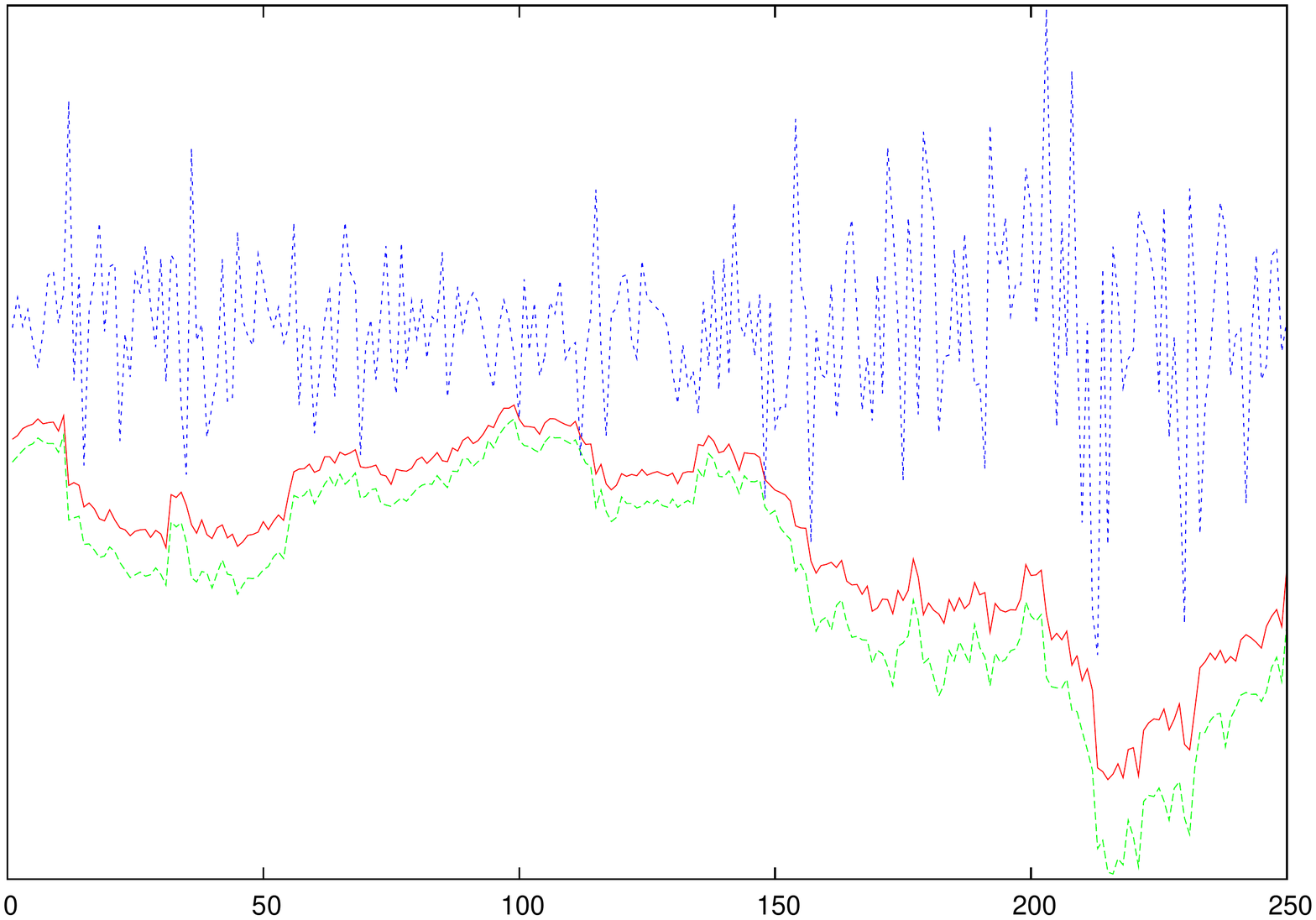}
\end{center}


\end{document}